\documentclass[a4paper]{article} 
\usepackage[utf8]{inputenc}

\usepackage[margin=1in]{geometry}

\usepackage{graphicx} 

\usepackage{natbib}

\usepackage[draft]{todonotes}
\usepackage{eurosym}

\usepackage{amsmath}
\usepackage{amsthm}
\usepackage{amssymb}
\usepackage{comment}

\usepackage{enumitem}

\usepackage{thmtools}
\usepackage{thm-restate}
\usepackage[hidelinks]{hyperref}
\usepackage[nameinlink]{cleveref}
\usepackage{url}

\allowdisplaybreaks

\global\long\def\lev{\operatorname{lev}}%
\global\long\def\fromto#1#2{#1,\, \ldots, #2}%

\global\long\def\setFromTo#1#2{\left\{  \fromto{#1}{#2}\right\}  }%

\global\long\def\argmax{\operatorname{arg\;max}}%

\global\long\def\Then{\Rightarrow}%

\global\long\def\blank{{\,\mathord{\cdot}\,}}%

\global\long\def\eps{\varepsilon}%

\global\long\def\norm#1{\lVert#1\rVert}%

\global\long\def\poly{\operatorname{poly}}%

\global\long\def\de{\mathrm{d}}%

\global\long\def\sq#1#2{\pars{#1}_{#2}}%

\global\long\def\parsg#1#2#3{\mathchoice%
  {\left#1 #3 \right#2}%
  {#1 #3 #2}%
  {#1 #3 #2}%
  {#1 #3 #2}%
}
\global\long\def\parstex#1{\parsg(){#1}}

\global\long\def\pars#1{\parstex{#1}}%

\theoremstyle{remark}

\theoremstyle{plain}
\newtheorem{prop}{\protect\propositionname}
\theoremstyle{plain}

\theoremstyle{plain}
\newtheorem{thm}{\protect\theoremname}
\theoremstyle{definition}
 \newtheorem{example}{\protect\examplename}
\theoremstyle{definition}
\newtheorem{defn}{\protect\definitionname}
\theoremstyle{plain}
\newtheorem{cor}{\protect\corollaryname}

\providecommand{\corollaryname}{Corollary}
\providecommand{\definitionname}{Definition}
\providecommand{\examplename}{Example}
\providecommand{\lemmaname}{Lemma}
\providecommand{\propositionname}{Proposition}
\providecommand{\remarkname}{Remark}
\providecommand{\theoremname}{Theorem}

\newcommand\nolinkemail[1]{\nolinkurl{#1}}

\newcommand\ym{y_{-i}} 
\newcommand\Dil{D_i}
\newcommand\Eil{E_i}

\title{Equilibria and Convergence in Fire Sale Games\thanks{%
    Supported by DFG Research Group ADYN under grant DFG 411362735 and by SNF grant P2ZHP2\textunderscore{}187934.
}}
\author{Nils Bertschinger\thanks{Goethe University Frankfurt, Frankfurt Inst.\ of Advanced Studies, Germany. \nolinkemail{bertschinger@fias.uni-frankfurt.de}}%
  \and%
  Martin Hoefer\thanks{Goethe University Frankfurt, Germany. \nolinkemail{mhoefer@em.uni-frankfurt.de}}%
  \and%
  Simon Krogmann\thanks{Hasso Plattner Institute, Potsdam, Germany.  \nolinkemail{simon.krogmann@hpi.de}}%
  \and%
  Pascal Lenzner\thanks{Hasso Plattner Institute, Potsdam, Germany. \nolinkemail{pascal.lenzner@hpi.de}}%
  \and%
  Steffen Schuldenzucker\thanks{Goethe University Frankfurt, Germany. \nolinkemail{schuldenzucker@em.uni-frankfurt.de}}%
  \and%
  Lisa Wilhelmi\thanks{Goethe University Frankfurt, Germany. \nolinkemail{wilhelmi@em.uni-frankfurt.de}}%
}

\date{}

\begin{document}

\maketitle

\begin{abstract}
The complex interactions between
    algorithmic trading agents can have a severe influence on the functioning of our economy, as witnessed by recent banking crises and trading anomalies.
    A common phenomenon in these situations are \emph{fire sales}, a contagious process of asset sales that trigger further sales.
    We study the existence and structure of equilibria in a game-theoretic model of fire sales. We prove that for a wide parameter range (e.g., convex price impact functions), equilibria exist and form a complete lattice. This is contrasted with a non-existence result for concave price impact functions.
    Moreover, we study the convergence of best-response dynamics towards equilibria when they exist. In general, best-response dynamics may cycle. However, in many settings they are guaranteed to converge to the socially optimal equilibrium when starting from a natural initial state. Moreover, we discuss a simplified variant of the dynamics that is less informationally demanding and converges to the same equilibria. We compare the dynamics in terms of convergence speed.
\end{abstract}

\section{Introduction}

On May 6, 2010, 2:45pm, one trillion dollars in stock market valuation
disappeared. In an event known as a \emph{flash crash}, the Dow Jones
and many other stock indices collapsed by as much as 9\%. The flash crash is generally
seen as the product of a system of interacting agents, many of them
computerized, that jointly exacerbated an initial shock.\footnote{The \citet{sec2010flashcrash} stated that the flash crash was triggered by a single market participant employing an (arguably simplistic) trading algorithm. Other trading agents at first absorbed the resulting price pressure, but then amplified and spread it. Other examples of flash crashes are extreme movements in currency markets in recent years~\citep{bis2016sterling,reuters2019bump}. Anecdotal evidence suggests that smaller-scale flash crashes happen very frequently~\citep{cnn2013mini,reuters2019bump}.
}
While prices recovered after approximately 30 minutes (after trading was briefly halted), it is by no means guaranteed that this will always be the case in a flash crash. Recent price crashes in cryptocurrency markets, such as the Bitcoin crash on April 17, 2021 \citep{forbes2012crypto}, may have resulted in permanently altered market
 conditions.
 
As prices deteriorate very quickly, it is important to understand
the amplification of stock price declines through the
interaction between different (electronic or human) trading agents.
In this paper, we study this problem through the lens of algorithmic
game theory.

Agents may amplify price deteriorations through many different processes. In this work, we focus on \emph{price-mediated contagion due to
leverage caps}. We study a collection of rational agents
that interact with each other through overlapping portfolios.
Each agent holds a share of the supply of one of several assets,
where usually
two or more agents hold a share in the same asset. Sales in any
given asset depress its price, i.e., sales have \emph{price
impact}, which in turn may reduce the value of the asset holdings
of another agent. We assume that agents are constrained by an upper
bound on their \emph{leverage}, i.e., the ratio between their (risky)
asset holdings and their equity. The equity of an agent is the difference between the total value of her assets and her liabilities, and it includes cash proceeds from the liquidation of risky assets. Leverage caps can represent the agents'
own desire to limit their risk, or they might be regulatory constraints.\footnote{The Basel~III regulatory framework stipulates leverage
caps for banks~\citep{basel2011basel}.}
Because of leverage caps, a price depression in some of the assets may
cause agents to perform further sales to satisfy their leverage constraints.
This can give rise to a contagious \emph{fire sale} process, where a small initial price drop quickly leads to a large number of asset sales and a corresponding price drop.

\subsection{Related Work}
Fire sales are a well-known phenomenon and have been studied
both academically and by regulators (e.g., central banks).
A \emph{leverage cycle}, in which banks' reduction of leverage leads to price decline and further reduction of leverage, was first studied by \citet{geanakoplos2009leverage}.
\citet{aymanns2015dynamics}
and \citet{aymanns2016taming} studied a dynamic model of leveraged and unleveraged investors.
\citet{cont2017fire}
studied a model of overlapping portfolios and leverage constraints.
Agents react to an initial shock following a specific iterative liquidation
process, where they proportionally sell off a fraction of their portfolio.
Our paper re-interprets a generalization of their model in
a game-theoretic setting. Importantly, \citet{cont2017fire} did not
study game-theoretic equilibria.
In a later work \citep{cont2019monitoring},
the same authors discuss risk indicators that help quantify the exposure of a given institution to price-mediated
contagion.
\citet{baes2020reverse} studied worst-case scenarios when agents respond to price depreciations in an individually optimal way.
\citet{banerjee2021price} studied a fire sale process where liquidations
occur at volume-weighted average prices.
Price-mediated contagion has also received interest from regulators, e.g., in the European Central Bank's STAMP\euro{} macro stress
testing framework \citep[Section~12.2.1]{dees2017stampe}.

Another related model for studying contagion effects in financial networks was introduced by~\citet{elliott2014financial}. This general model considers banks and assets and allows for banks to own shares of other banks. Banks in the network are connected by linear dependencies, i.e., cross-holdings, and if any bank’s value drops below a critical threshold, its value suffers an additional failure cost, potentially impacting the value of other banks. \citet{hemenway2015sensitivity} study the sensitivity and computational complexity of this model.

Recently, there has been considerable interest in analyzing financial networks from an algorithmic and game-theoretic point of view. A number of works are based on a classical model for systemic risk in financial networks by \citet{eisenberg2001systemic} which studies the clearing problem, i.e., to determine which banks are in default and their exposure to systemic risk. A recent work by \citet{bertschinger2020strategic} proposes a strategic version, in which firms are rational agents in a given directed graph of debt contracts. To clear its debt, every agent strategically decides on a ranking-based payment strategy. The authors study the existence and computational complexity of pure Nash and strong equilibria, and provide bounds on the (strong) prices of anarchy and stability. \citet{KKZ21} consider the same model but they focus on more general priority-proportional payments. Recently, \citet{HW22} analyzed minimal clearing and the impact of seniorities (i.e., priorities over debt contracts) on the existence of equilibria.  
Moreover,~\cite{KKZ22} study complexity questions regarding centralized bailouts and the forgiving of certain debts for banks in default, along with game-theoretic incentives that emerge in such scenarios.

The network model by \citet{eisenberg2001systemic} has also been augmented by considering credit-default swaps (CDS)~\citep{schuldenzucker2016clearingOLD}. The hardness of finding clearing payments with CDS is analyzed by~\citet{schuldenzucker2017clearing} and the approximation perspective was considered recently by~\citet{IKV_ICALP22}. The impact of banks deleting or adding liabilities, donating to other banks, or changing external assets has been studied by~\cite{papp-icalp}. These operations can be beneficial for the individual banks since the changes may enforce more favorable solutions. Still, after clearing some banks may end up in default. \citet{PappW21} study the influence of the sequence of banks' defaults. In \citep{PappW21_WINE} several possible strategies for resolving default ambiguity, i.e., which banks end up in default and how much of their liabilities can these defaulting banks pay, are studied. In a different direction, in \citep{PappW21_EC} risk mitigation via local network changes (``debt swapping'') is investigated.

Frequent call markets represent another financial game which has sparked research recently  \cite{call-markets-1, call-markets-2}.
Here, the focus lies on preventing fraud while maintaining efficiency.

\subsection{Our Contribution}
In this paper, we innovate upon prior work by studying fire sales
as a static \emph{fire sale game} played by fully rational agents.
Each agent decides on the fraction of her portfolio she
sells and keeps, respectively.\footnote{%
    Note that the relative composition of each agent's portfolio is kept constant in all of this paper except for \Cref{sec:non-even}. This is a standard assumption in the literature on price-mediated contagion \cite[see, e.g.,][]{cont2017fire,duarte2021fire} and supported by empirical evidence \cite[e.g.,][]{girardi2021portfolio}.
    Our main reason to discuss this case is that it helps to simplify the presentation.
    Most of the results in this paper can be extended in a straightforward way to general monotone sales, see our discussion in~\Cref{sec:Discussion}.
}
The agent derives a utility equal to her equity, i.e., the shareholder value of her firm, as long as she satisfies the leverage constraint.\footnote{
    We assume that agents that cannot satisfy the leverage constraint need to fully liquidate their position.} 
The equity depends on her own action (through sales choices and
price impact) and the actions of others (through price impact). We
study the Nash equilibria of this game under different restrictions
on the strategy space, assumptions on the price impact function,
and assumptions on how the price impact manifests while selling any given asset. We capture this
latter dimension in a parameter $\alpha\in[0,1]$.

We study the existence and the structure of equilibria.
When $\alpha=1$ (i.e., price impact
affects prices of assets sold as much as those of assets kept)
or if the price impact function is convex, we show that an equilibrium
always exists and the set of equilibria further has a desirable \emph{lattice structure}.
In this case, a process of iterative best responses converges to the point-wise maximal equilibrium, which sets each agent off best among all equilibria (\Cref{sec:equilibrium-existence}).
In contrast, if $\alpha<1$ and price impact is concave (\Cref{sec:concavePrices}), an equilibrium need not even exist.

Under the assumptions that allow showing the above lattice structure result, 
the Pareto optima of the game form strong equilibria. This is not true for $\alpha<1$. In this case “bank-run” effects can manifest, where an agent is incentivized to sell more than strictly necessary to satisfy her leverage constraint. Moreover, even under severe restrictions, there are games where a Nash equilibrium yields devastating utility for every agent in comparison to the social optimum. 

Also, we study processes by which agents may actually converge to
an equilibrium, where we focus on the case of $\alpha=1$ and linear price impact.
We consider two kinds of best-response dynamics: 
the regular best-response dynamics and a simplified best-response
dynamics, where agents do not take the fact into account that their own response to changing market prices in turn generates price impact.
One may argue that the simplified dynamics serve as a more realistic
model of agent behavior since it is less complex to execute and requires
less information about the precise shape of the price impact function.
We show that, under the assumptions of our lattice structure result,
both of these dynamics converge to the maximal equilibrium.
However, we show via computational experiments that they can do so at vastly different speeds. Our experiments further suggest that the convergence speed depends on the overall diversification of agents across assets, with the worst case depending on the specific parameters (\Cref{sec:convergence-dynamics}).

Finally, in \Cref{sec:non-even} we briefly touch upon an extension of our model to a setting where agents are not restricted to selling off a share of their whole portfolio proportionally. Rather, they can choose which assets to sell to which degree. We show that the possibility of non-proportional sales leads to significantly more complex strategic interactions and less desirable outcomes. More in detail, even for $\alpha = 1$ and linear price impact the social-welfare optimum need not be an equilibrium and the price
of stability is greater than 1. 
Nevertheless, the single-agent best-response problem is still tractable.

Overall, our results illustrate that the equilibrium structure arising from a strategic interaction through fire sales is crucially dependent on the agents' strategy spaces and the way in which market impact manifests.
Coordination on an equilibrium is of central importance, suggesting potential for regulatory assistance.

\section{Preliminaries\label{sec:Preliminaries}}

\subsection{The Model\label{sec:Model}}
Our market model is based on \citet{cont2017fire}.
We have sets of \emph{agents} $N=\setFromTo 1n$ and \emph{assets}
$M=\setFromTo 1m$.
Each agent $i$ holds an amount $a_{i}^{I}>0$ of \emph{illiquid assets}, which cannot be sold and are not subject to price impact.
Agent $i$ holds an amount of $x_{ij}$ in the (liquid) asset $j$. We assume
w.l.o.g.\ the amounts to be normalized, so that $x_{ij}\in[0,1]$ for all $i,j$, if no agent sells any holdings in $j$.

Each agent $i$'s strategic action consists of a number $y_{i}\in [0,1]$,
which is the share of agent $i$'s holdings in each asset
$j$ that agent $i$ \emph{keeps}. Agent $i$ thus sells an amount of $1-y_{i}$ of her total portfolio holdings on the market. The amount of asset
$j$ held by agent $i$ after selling is $y_{i}x_{ij}$. Let
$x_{j}(y):=\sum_{i\in N}y_{i}x_{ij}$ be the amount of asset $j$
that has not been sold if agents act according to $y$.
We assume that the price of each asset $j$ decays when assets are
sold according to a function $p_{j}(y)=p_{j}(x_{j}(y))$ such that
$p_{j}(0)=0$, $p_{j}(1)=p_{j}^{0}$, and $p_{j}(x_{j})$ is continuous
and increasing in $x_{j}$. 
We say that price impact is \emph{linear}
if for each $j\in M$, the function $p_{j}$ simply linearly interpolates
between the two given points, i.e., $p_{j}(x_{j})=x_{j}p_{j}^{0}$. We
say that price impact is \emph{convex} if $p_{j}$ is convex for each
$j$ and \emph{concave} if $p_{j}$ is concave for each $j$. If all
agents sell according to $y$, the value of agent $i$'s remaining
asset holdings is therefore 
\[
a_{i}(y)=a_{i}^{I}+y_{i}\sum_{j\in M}x_{ij}p_{j}(y)
.
\]

Selling assets redeems a certain amount of money for each agent $i$.
As the agents sell their assets, the corresponding price impact would usually manifest
over time: amounts of assets that are sold at the very beginning would
typically not be subject to price impact, while sales that happen
late would bear significant price impact.
Since we consider a static game, we do not model this effect directly.
Instead, we follow \citet{cont2017fire} by capturing the effect using an \emph{implementation shortfall parameter} $\alpha\in[0,1]$. Agents redeem a
share of $\alpha$ of their sales according to the post-price-impact
price and a share of $1-\alpha$ according to the pre-price-impact
price of an asset. If $\alpha=0$, the market reacts very slowly to
asset sales, so that price impact does not manifest in the effective
price that agents receive when they sell their assets (but it does manifest for the post-sale values of the remaining
assets). If $\alpha=1$, then the full price impact manifests immediately;
such a situation may arise when asset sales are conducted using an
auction mechanism.\footnote{An alternative approach is to compute the integral of
of
the price impact function with respect to a path of infinitesimal
trades. This route was chosen by \citet{banerjee2021price}. For the
purpose of this paper, where we are interested in the strategic implications
of fire sales, our approach using the $\alpha$ parameter provides
a simple way of capturing the extent to which price impact affects agents.} If all agents sell according to $y$, the total revenue that agent
$i$ derives from her asset sales is now 
\[
\Delta_{i}(y)=(1-y_{i})\sum_{j\in M}x_{ij}((1-\alpha)p_{j}^{0}+\alpha p_{j}(y)).
\]
We assume that each agent $i$ has \emph{liabilities} of
$l_{i}$ to external creditors. Agent $i$'s \emph{equity} is the
difference between her total assets and liabilities, where her total assets consist of her (illiquid and liquid) assets and the risk-free money she has redeemed from asset sales:
\[
e_{i}(y)=a_{i}(y)+\Delta_{i}(y)-l_{i}.
\]
Note that the equity of an agent is what would remain if the agent's
(risky and risk-free) assets were used to pay off her liabilities.
It therefore equals the agent's shareholder value.

If $e_{i}(y)>0$, then agent $i$'s \emph{leverage} at $y$ is the
ratio between her risky assets and her equity:
\[
\lev_{i}(y)=\frac{a_{i}(y)}{e_{i}(y)}=\frac{a_{i}(y)}{a_{i}(y)+\Delta_{i}(y)-l_{i}} \enspace.
\]
Note how an agent that holds no risky assets (i.e., $a_i(y)=0$) has a leverage of 0 (unless
its equity is also 0) while an agent that holds high risky assets,
only little risk-free assets and has high liabilities (i.e., $a_i(y)$ and $l_i$ are large and $\Delta_i(y)$ is small) has a
high leverage. This is why leverage is used as an instrument
to gauge the riskiness of an agent. If $e_{i}(y)\le0$, then $\lev_{i}(y)$
is not defined.

We assume that regulatory constraints limit the admissible leverage
of an agent to a constant $\lambda>1$, i.e., agent $i$ needs to choose
its action $y_{i}$ such that
\[
\lev_{i}(y)\le\lambda.
\]
If no such $y_{i}$ exists, we say that agent $i$ is \emph{illiquid}
given the actions $y_{-i}$ of the other agents. If no $y_{i}$ exists
for which $e_{i}(y)>0$, we say that agent $i$ is \emph{insolvent}
at $y_{-i}$. Insolvent or illiquid agents need to sell their whole
asset holdings (otherwise, we assume that they receive utility $-\infty$); the other agents (which we call \emph{liquid} agents)
derive a utility equal to their equity.
More in detail, we consider the following utility function.
Define special strategies $y_i^0 := 0$ and $y_i^1 := 1$ and define the \emph{utility} of agent $i$ as
\[
u_i(y) :=\begin{cases}
-\infty & \text{if } y_i \neq y_i^0 \text{ and }
\\&\hspace{0.3cm}
\left(e_i(y) \le 0 \text{ or } \lev_i(y) > \lambda\right)
\\
e_i(y) & \text{otherwise.}
\end{cases}
\]
Note that agents always have the option to sell everything (i.e., play $y_i=y_i^0$) and then receive a utility equal to their equity.
This is motivated by the fact that it should always be possible to liquidate a firm; regulation must not prevent agents from exiting the market.
Observe that, in any Nash equilibrium, insolvent or illiquid agents sell
everything and liquid agents
either play a $y_{i}$ for which $e_{i}(y)>0$ and $\lev_{i}(y)\le\lambda$ or sell everything (i.e., play $y_i=y_i^0$).
We call a collection $(N,M,a^{I},x,\alpha,\lambda)$ a \emph{fire
sale game}. See \Cref{fig:intro_example} for an example instance.

\begin{figure*}[ht!]
    \centering
    \includegraphics[width=\linewidth]{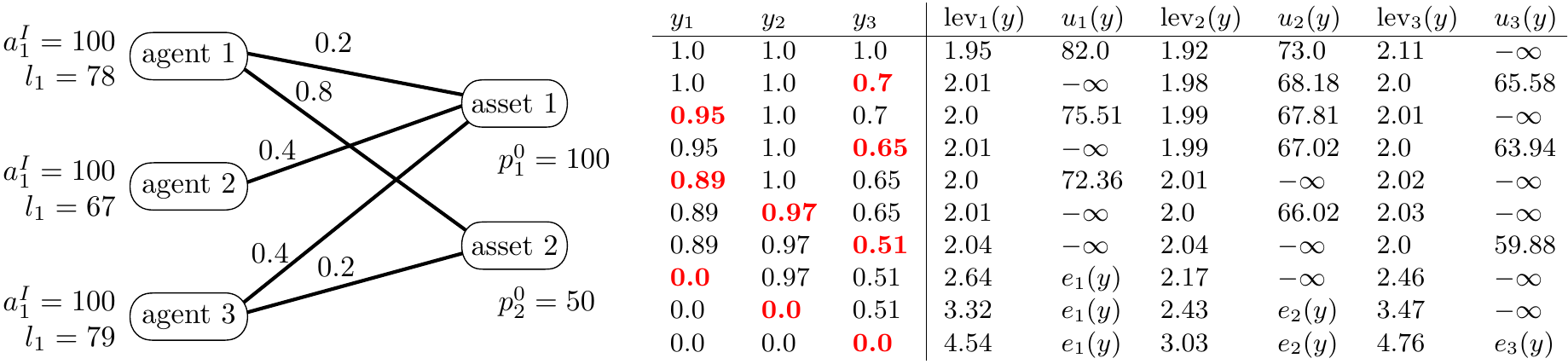}
    \caption{A fire sale game with three agents and two assets, $\alpha=1$, $\lambda=2.0$, and linear price impact, i.e., $p_j(y) = p_j^0 \cdot \sum_{i\in N} x_{ij}y_i$. Left: Asset holdings and initial prices. Right: Sequence of game states obtained by playing best responses starting from the state $(1,1,1)$. Newly chosen strategies are red, all values are rounded to two digits. When Agent 3 starts selling to fulfill her leverage constraint, a fire sale starts that eventually forces all agents to sell everything.}
    \label{fig:intro_example}
\end{figure*}

\subsection{Basic Properties\label{sec:Basic-Properties}}

Our first proposition shows that sales of one agent destabilize other agents with overlapping portfolios, in the sense that their leverage increases. This effect introduces
fire sales into our model. 
For technical reasons, we need to consider a lower bound of 1 on the leverage; note that this is irrelevant to our discussion since the leverage cap is always $\lambda > 1$.

\begin{restatable}{prop}{propone}\label{prop:leverage_inc}
Let $\overline{\lev}_i(y) = \max(1, \lev_i(y))$. Then $\overline{\lev}_i(y_i, y_{-i})$ is monotonically decreasing in $y_{-i}$. More in detail, if $y_i \in [0,1]$ and $\ym \le \ym'$ point-wise, then $\overline{\lev}_i(y_i, \ym) \ge \overline{\lev}_i(y_i, \ym')$.
\end{restatable}
\begin{proof}
By assumption and monotonicity of the functions $p_j$, we have $p_j(y_i, \ym) \le p_j(y_i, \ym')$ for all $j$.

It immediately follows from the definition that $\lev_i(y_i, \ym) < 1$ if and only if $\Delta_i(y_i, \ym) > l_i$. If this is the case, then also $\Delta_i(y_i, \ym') \ge \Delta_i(y_i, \ym) > l_i$, where the first inequality immediately follows from the above statement about the prices $p_j$. Now also $\lev_i(y_i, \ym') < 1$ and thus $\overline{\lev}_i(y_i, \ym) = 1 = \overline{\lev}_i(y_i, \ym')$.

Assume now that $\lev_i(y_i, \ym),\, \lev_i(y_i, \ym') \ge 1$ and write short $a:=a_i(y_i,\ym)$, $a':=a_i(y_i,\ym')$, and likewise for $\Delta$. Then
\begin{align*}
    \overline{\lev}_i(y_i,\ym) &= \lev_i(y_i, \ym)
    = \frac a {a + \Delta -l_i} \ge \frac {a'} {a' + \Delta - l_i}
    \ge \frac {a'} {a' + \Delta' - l_i}
    = \lev_i(y_i, \ym') = \overline{\lev}_i(y_i, \ym')
    ,
\end{align*}
where the first inequality holds since, by assumption, $\Delta - l_i \le 0$ and $a \le a'$, and the second inequality holds since $\Delta \le \Delta'$, both of which immediately follow from monotonicity of the prices $p_j$.
\end{proof}

Note that $\lev_i$ is \emph{not} necessarily monotonically increasing in $y_i$, i.e., selling more does \emph{not always} reduce leverage. Whether or not this is the case depends on the price impact functions, the $\alpha$ parameter, and $\ym$.

The leverage function $\lev_i$ is also continuous in $y$.

\begin{restatable}{prop}{proptwo}
\label{prop:cont_lev_all_sales} The
leverage function $\lev_{i}$ is continuous in the strategy profile
$y$ in the region where agent $i$ is solvent.
\end{restatable}
\begin{proof}
This follows immediately from the definitions. More in detail, continuity of the price impact functions $p_j$ implies continuity of the functions $a_i$ and $\Delta_i$, which in turn imply continuity of $\lev_i$ on the domain where its denominator is non-zero.
\end{proof}

Agents become illiquid before they become insolvent. This technical property will be useful in our proofs.
\begin{restatable}{prop}{propthree}\label{prop:illiquid-before-insolvent}
Let $i$ be an agent and let $(y^t)_t$ be a sequence of strategy profiles
such that $e_i(y^t) > 0$ for all $t$ and $e_i(y^t)\to 0$. Then $\lev_i(y^t)\to \infty$. In particular, there exists a $t$ such that $\lev_i(y^t) > \lambda$.
\end{restatable}
\begin{proof}
This follows immediately from the definition $\lev_i(y) = a_i(y)/e_i(y)$, the assumption $e_i(y^t)\to\infty$, and the fact that $a_i(y^t) \ge a_i^I > 0$ for all $t$.
\end{proof}

\subsection{Post-Sale Prices}
\label{sec:basic-post-prices}

An important special case of the model is when $\alpha=1$. Here, agents liquidate their assets at an average price that is equal to the price after all assets have been sold.
We also say that in this case the agents receive \emph{post-sale prices}.
Agents are highly affected by price devaluations in this case.
Intuitively,
agents will therefore avoid sales and only execute them to satisfy
the leverage constraint.
The following \Cref{prop:monotonic_post_br-fkt} formalizes this intuition.
We will show that this further implies monotonicity of the best response of each agent, which will be an important building block towards our later results.

The best-response function $\Phi : [0,1]^{N} \to [0,1]^{N}$ is
\begin{align*}
\Phi_i(y) & :=\argmax_{y_{i}\in[0,1]}u_{i}(y_{i},y_{-i}).
\end{align*}
Ties in the $\argmax$ are w.l.o.g.\ broken in favor of largest
$y_{i}$.

It is easy to see that for $\alpha=1$, the equity $e_i$ simplifies to
\begin{equation}
    e_i(y) = a_i^I - l_i + \sum_j x_{ij} p_j(y).
    \label{eq:post-sale-equity}
\end{equation}
The equity is equal to the assets, net of liabilities, under the assumption that the agent has not actually sold anything, but is still exposed to price impact on her asset holdings. This drives the following result.

\begin{restatable}{prop}{propfour}
\label{prop:monotonic_post_br-fkt} 
Let $\alpha=1$. Then (1) each liquid agent $i$ maximizes her utility by maximizing $y_{i}$ subject to the leverage constraint and (2) $\Phi$ is monotonic.
\end{restatable}
\begin{proof}
First, observe that the equity is monotonic in $y_{i}$ since the functions $p_j$ are monotonic and by \eqref{eq:post-sale-equity}.
The best response of agent~$i$ thus minimizes sales (i.e., maximizes $y_{i}$) subject to the leverage constraint.

Now consider the best response $y_{i}^{*}$ of $i$ to a fixed $y_{-i}$. First, assume the special case where $y_{i}^{*}=0$. Then, by monotonicity of the equity, agent $i$ either exactly satisfies the leverage constraint or is illiquid or insolvent. In the latter scenario, the leverage is undefined, while the other cases imply $\lev_i(y_i^*, \ym) \geq \lambda > 1$. By monotonicity of the equity and the leverage (\Cref{prop:leverage_inc}) in $\ym$, the agent must maintain this strategy when the others increase their sales. Now assume that $y_{i}^{*}>0$. To maximize utility, the agent chooses the highest value $y_{i}$ as her strategy so that the leverage constraint is still satisfied.
Whenever the agent fulfills the leverage constraint without liquidating any assets (i.e., $y_i^* = 1$), her leverage is at least 1, since $\lev_i(y_i^*, y_{_i})=a_i(y_i^*,y_{-i})/(a_i(y_i^*,y_{-i})-l_i)\geq 1$. Now assume the agent to sell a share of her assets, i.e., $y_i^*\in (0,1)$. In particular, this means that $\lev_{i}(y_{i}^{*},y_{-i})=\lambda>1$ and $\lev_{i}(\hat{y}_{i},y_{-i})>\lambda$ for all $\hat{y}_{i}>y_{i}^{*}$. If the other agents now increase their sales, all $\hat{y}_{i}$ remain invalid (i.e., $y_{-i}'\le y_{-i}\Then\lev_{i}(\hat{y}_{i},y_{-i}')>\lambda$ for all $\hat{y}_{i}>y_{i}^{*}$, due to \Cref{prop:leverage_inc}). Therefore, if $y_{i}^{*}$ no longer satisfies the constraint, $i$ must also liquidate more assets. 
\end{proof}

\subsection{Intermediate and Convex Prices}
Another case that we will show is relatively well-behaved is when the price function $p_j(x_j)$ is convex, even if $\alpha \in (0, 1)$. In this situation, agents receive a payoff somewhere between the pre- and post-sale prices, and any trade carries a higher price impact when little is sold overall compared to when a lot of assets are already being sold.

When only few assets are being sold and prices are still high, we show that agents have an incentive to sell as little as possible while maintaining their leverage constraint, like in \Cref{sec:basic-post-prices}. Intuitively, by convexity of the price impact function and when the other agents only sell little, any additional sale by agent $i$ carries a large price impact and is thus often undesirable. Once a sufficiently high amount of assets is being sold by the other agents, however, the price impact of any additional unit sold becomes small enough that it becomes profitable for an agent $i$ to “run” on the market and sell all of their asset holdings.\footnote{%
    Depending on the game, either of the two “phases” may be absent. The point at which “running” on the market becomes profitable is also agent-specific and multi-dimensional. We express this in terms of a monotonicity statement regarding the vector $\ym$ of the strategies of the other agents in \Cref{prop:monotonic_convex_br-fkt}.
}
As an important technical tool for formalizing this intuition, we show that the equity of an agent is convex in her own strategy.

\begin{restatable}{prop}{propfive}
\label{prop:p-convex-e-convex}
Let $y_{-i}$ be fixed. With convex price impact, the equity $e_{i}(y_{i},y_{-i})$
is convex in $y_{i}$.
\end{restatable}
\begin{proof}
We prove convexity
of $e_{i}(y_{i},y_{-i})$ by showing that the first derivative $e'_{i}(y_{i},y_{-i})$ with respect to $y_i$
is monotonically increasing. Let $\hat{y}_{i}>y_{i}$, then by monotonicity of all price functions $p_{j}$ and derivatives $p'_{j}$
we have 
\begin{align*}
e'_{i}(y_{i},y_{-i}) & = \sum_{j} x_{ij} ((\alpha-1)p_{j}^{0}+(1-\alpha)p_{j}(y_{i},y_{-i}) + (y_{i}+\alpha-y_{i}\alpha)p'_{j}(y_{i},y_{-i}))\\
 & \leq \sum_{j} x_{ij} ((\alpha-1)p_{j}^{0}+(1-\alpha)p_{j}(\hat{y}_{i},y_{-i}) + (y_{i}+\alpha-y_{i}\alpha)p'_{j}(y_{i},y_{-i}))\\
 & \leq \sum_{j} x_{ij} ((\alpha-1)p_{j}^{0}+(1-\alpha)p_{j}(\hat{y}_{i},y_{-i}) + (\hat{y}_{i}+\alpha-\hat{y}_{i}\alpha)p'_{j}(\hat{y}_{i},y_{-i}))\\
 & =e'_{i}(\hat{y}_{i},y_{-i}).\qedhere
\end{align*}
\end{proof}

\Cref{prop:monotonic_post_br-fkt} states that the best-response function is monotonic if $\alpha=1$. In case of convex price impact, monotonicity
of the best-response function also follows without any restriction on $\alpha$. Convexity of the equity gives rise to the aforementioned dichotomy of potential optimal strategies: either sell as little as possible or sell everything. Note that, different from the $\alpha=1$ case, the equity is \emph{not} necessarily monotonic in $y_i$ for convex price impact.

\begin{restatable}{prop}{propsix}
\label{prop:monotonic_convex_br-fkt}
Consider convex price impact.
Then the best-response function $\Phi$
is monotonic. 
\end{restatable}
\begin{proof}
\newcommand\yimax{y_i^{\max}}
\newcommand\yimaxp{y_i^{\prime{\max}}}
By convexity of the equity function, agent $i$ maximizes her utility
by either minimizing or maximizing her sales. For every strategy profile
$y$, it is therefore sufficient to consider the maximal and minimal
feasible strategy (the minimal one being 0) as a potential best response. Let $y_{i}^{\max}$ be the maximal
strategy fulfilling the leverage constraint for fixed $y_{-i}$, or 0 if no such strategy exists. Then,
the best-response function is
\[
\Phi_{i}(y_{i},y_{-i})=
\argmax_{y_i \in \{y_i^{\max},y_i^0\}} u_i(y_i, \ym),
\]
where a potential tie is broken in favor of $y_i^{\max}$.
The value $y_i^{\max}$ is monotonic as a function of $\ym$, which follows from monotonicity of $\lev_i$ in $\ym$ as in the proof of \Cref{prop:monotonic_post_br-fkt}.
Trivially, $y_i^0$ is monotonic in $\ym$ as well.
It remains to show that the choice between the two cases $y_i^0$ and $\yimax$ is sufficiently monotonic in $\ym$.

To this end, let $\ym \le \ym'$ and let $\yimaxp$ be the value of $\yimax$ corresponding to $\ym'$. We show that the above form for $\Phi_i$ implies that $\Phi_i(\ym) \le \Phi_i(\ym')$. If agent $i$ is insolvent or illiquid at $\ym$, then $\yimax=y_i^0=0$ and the statement is trivial.
So assume that $i$ is liquid under $\ym$. By monotonicity of $\lev_i$ (\Cref{prop:leverage_inc}) then $i$ is also liquid under $\ym'$.
It remains to be shown that if
$e_i(y_i^0,y_{-i}) \leq e_i(\yimax,y_{-i})$, then also $e_i(y_i^0,\ym') \leq e_i(\yimaxp, \ym')$.
We proceed in two steps.
Note that $\yimax$ is feasible for agent $i$ at $\ym'$ by monotonicity of the leverage function.
\begin{enumerate}[label=(\roman*)]
    \item $e_i(y_i^0,\ym') \leq e_i(\yimax, \ym')$:
    Define
    \begin{align*}
    \delta_{j}^{0}&=p_{j}(y_{i}^{0},\ym')-p_{j}(y_{i}^{0},y_{-i})\\
    \delta_{j}^{\max}&=p_{j}(\yimax,\ym')-p_{j}(\yimax,y_{-i}).
    \end{align*}
    Observe that $\delta_j^0, \delta_j^{\max} \ge 0$ by monotonicity and
    $\delta_j^0 \leq \delta_j^{\max}$ by convexity and monotonicity of the price function~$p_j$.\footnote{%
        Recall that each price impact function $p_j$ is a one-dimensional function $[0, 1] \to [0,1]$ and we use the short notation $p_j(y) := p_j(x_j(y))$, where $x_j(y) = \sum_i x_{ij} y_i$.
    }
    It is easy to see from the definition that
    \begin{align*}
        a_i(\yimax, \ym') =& a_i(\yimax, \ym) + \yimax \sum_j x_{ij} \delta_j^{\max}\\
        \Delta_i(\yimax,\ym') =& \Delta_i(\yimax,\ym') + (1-\yimax) \alpha \sum_j x_{ij} \delta_j^{\max}
    \end{align*}
    and analogously for $y_i^0$.
    We therefore receive
    \begin{align*}
        e_i(\yimax, \ym') =& a_i(\yimax, \ym') + \Delta_i(\yimax, \ym') - l_i
        \\
        =& a_i(\yimax, \ym) + \Delta_i(\yimax,\ym) - l_i + (\yimax + (1-\yimax)\alpha) \sum_j x_{ij}\delta_j^{\max}
        \\
        =& e_i(\yimax, \ym) + (\alpha + (1-\alpha) \yimax) \sum_j x_{ij}\delta_j^{\max}\\
        e_i(y_i^0, \ym') =& e_i(y_i^0, \ym) + \alpha \sum_j x_{ij}\delta_j^0
        .
    \end{align*}
    The statement now follows because, by assumption, $e_i(\yimax, \ym) \ge e_i(y_i^0, \ym)$ and furthermore $\delta_j^{\max} \ge \delta_j^0 \ge 0$.

    \item $e_i(y_i^0,\ym') \leq e_i(\yimaxp, \ym')$:
    By convexity of $e_i(\blank,\ym')$, this function takes on its maximal value at one of the two extreme feasible values $0$ or $\yimaxp$. By (i), this is not at $y_i^0$ (up to tie-breaking, which is decided in favor of $\yimaxp$). Therefore, it must be at $\yimaxp$.\qedhere
\end{enumerate}
\end{proof}

\section{Equilibrium Existence and Convergence of Dynamics\label{sec:equilibrium-existence}}

In this section, we discuss the existence of equilibria in a given fire sale game and the convergence of best-response dynamics to equilibrium. Our first main result is that, in the above-discussed cases where the best-response function is monotonic, the set of equilibria has
a particularly desirable \emph{lattice structure}.
This in particular implies that (1) an equilibrium always exists and
(2) there is always an equilibrium that minimizes the sales of each asset by each agent simultaneously and maximizes
the equity of each individual agent among all equilibria. This equilibrium
is in particular a Pareto optimum and maximizes social welfare among all equilibria. We
first state this result in abstract terms in \Cref{lem:mon-phi-lattice}.
The lemma together with \Cref{prop:monotonic_post_br-fkt} yields the main structural result.

\begin{restatable}{lem}{lemmaone}
\label{lem:mon-phi-lattice}
Assume that the best-response function
$\Phi$ is monotonic. Let $E$ be the set of Nash equilibria. Then $E$ is non-empty, and
the pair $(E,\geq)$ forms a complete lattice.
\end{restatable}

\begin{proof}
The set of all strategy profiles is $D=[0,1]^{n\times m}$,
and $(D,\geq)$ is a complete lattice. The map $\Phi:D\rightarrow D$
computes, for a given strategy profile $y$, the best response $y_{i}^{*}$
for every agent $i$ with respect to $y_{-i}$. Thus, $\Phi(y)=(y_{1}^{*},y_{2}^{*},\dots,y_{n}^{*})$ 
equals the strategy profile after every agent deviated to her best
response simultaneously. The fixed points of the map are the Nash
equilibria of the fire sale game. By assumption, $\Phi$ is monotonic.
Applying the Knaster-Tarski Theorem, the statement follows. 
\end{proof}

\Cref{lem:mon-phi-lattice} together
with Propositions~\ref{prop:monotonic_post_br-fkt} and~\ref{prop:monotonic_convex_br-fkt} yield:

\begin{thm}
\label{thm:lattice}
Assume that (1) $\alpha=1$ or (2) $p_{j}$ is
convex for all $j$. Let $E$ be the set of Nash equilibria. Then $E$ is non-empty, and
the pair $(E,\geq)$ forms a complete lattice.
\end{thm}

Under the assumptions of the above theorem, iterating the best-response function further converges to the point-wise maximal equilibrium.

\begin{restatable}{thm}{theoremtwo}
\label{thm:br-convergence}
Assume that (1) $\alpha=1$ or (2) $p_{j}$ is
convex for all $j$. 
Let $(y^{t})$ be the iteration sequence defined as $y^{1}=(1,1,\dots,1)$ and $y^{t+1}=\Phi(y^{t})$.
Then $(y^t)$ converges to the point-wise maximal equilibrium.
\end{restatable}

\begin{proof}
\renewcommand\phi\varphi 
The statement follows from the fact that $\Phi$ is monotonic and continuous from above.
This is a standard technique and can be seen as a special case of the Kleene fixed point theorem (see, e.g., \cite[Lemma 3]{schuldenzucker2017default}
for a direct proof).
Monotonicity follows from Propositions~\ref{prop:monotonic_post_br-fkt} and~\ref{prop:monotonic_convex_br-fkt}. We show continuity from above.

Given $y_{-i}$, let $\Dil(y_{-i}) = \{y_i \mid e_i(y_i, y_{-i}) > 0 \wedge \lev_i(y_i, y_{-i}) \le \lambda\}$.
Let $E_i^\lambda = \{y_{-i} \mid \Dil(\ym) \neq \emptyset\}$.
\Cref{prop:illiquid-before-insolvent} implies
that the set $\{y \mid e_i(y) > 0 \wedge \lev_i(y) \le \lambda\}$ is closed. This shows that $\Dil(y_{-i})$ for any $y_{-i}$ and $E_i^\lambda$ are closed because they are projections of this set (all involved sets are bounded, so being closed and being compact are equivalent).
$\Dil(\,\cdot\,)$ is monotonic in the sense that if $\ym \le \ym'$ point-wise, then $\Dil(\ym) \subseteq \Dil(\ym')$. This follows from monotonicity of $e_i$ and inverse monotonicity of $\lev_i$ in $\ym$ (where $\lev_i(y)>1$, cf.\ \Cref{prop:leverage_inc}).
By the same argument, $\Eil$ is monotonic in the sense that if $\ym \le \ym'$ and $\ym \in \Eil$, then $\ym' \in \Eil$.

In case (1) with $\alpha=1$ \Cref{prop:monotonic_post_br-fkt} yields
$\Phi_{i}(y)=
\varphi_i(y_{-i}) := \max \Dil(y_{-i}) \text{ if } y_{-i} \in E_i^\lambda$ and $0$ otherwise.

Let now $(\ym^t)_t$ be any point-wise decreasing sequence in $[0,1]^{n-1}$ and let $\ym^+ = \lim_t \ym^t$. By monotonicity and closedness of $\Eil$, we know that, if $\ym^+ \notin \Eil$, then $\ym^t \notin \Eil$ for almost all $t$, so $\Phi_i(\ym^t)=0$ for almost all $t$, and we have $\lim_t \Phi_i(\ym^t) = \Phi_i(\ym^+)$. So assume $\ym^+ \in \Eil$ and, thus, $\ym^t \in \Eil$ for all $t$, since $\ym^t \ge \ym^+$. It remains to prove that $\phi_i(\ym^+) = \lim_t \phi_i(\ym^t)$.

By continuity of $\lev_i$, we have $\lev_i(\lim_t \phi_i(\ym^t), \ym^+) = \lim_t \lev_i(\phi_i(\ym^t), \ym^t) \le \lambda$, where the inequality is by definition of $\phi_i$. Thus, $\lim_t\phi_i(\ym^t)\in \Dil(\ym^+)$ and thus, by choice of $\phi_i(\ym^+)$, we have $\lim_t\phi_i(\ym^t) \le \phi_i(\ym^+)$.
On the other hand, by monotonicity of $\lev_i$ in $\ym$, $\Dil(\ym^+)\subseteq \Dil(\ym^t)$ for all $t$ and thus, by choice of the $\phi_i(\blank)$, we have $\phi_i(\ym^+) \le \lim_t \phi_i(\ym^t)$.
We obtain equality as required.

Consider case (2), where $\alpha \in (0, 1)$ and $p_j$ is convex. From the proof of \Cref{prop:monotonic_convex_br-fkt}, we know the following:
\begin{itemize}
    \item $\Phi_i(\ym) \in \{0,\,\phi_i(\ym)\}$ if $\ym \in \Eil$.
    \item $\Phi_i(\ym) = 0$ if $\ym \notin\Eil$.
    \item If $\Phi_i(\ym) = 0$ and $\ym' \le \ym$, then also $\Phi_i(\ym')=0$.
\end{itemize}
Hence, if it is ever the case that $\Phi_i(\ym^t)=0$, then this will be the case for almost all $t$ and for $\ym^+$, and thus convergence holds trivially.
If this is never the case, then our above argument regarding continuity of $\phi_i$ applies unmodified and implies the statement of the theorem.
\end{proof}

Equilibrium existence and convergence do not hold for $\alpha\in(0,1)$ and \emph{concave} price impact. In \Cref{sec:concavePrices} we describe games in which no equilibrium exists.

In \Cref{thm:br-convergence}, we assume that agents \emph{concurrently} deviate to best responses, i.e., $\Phi(y^t)$ applies best responses simultaneously to each component of the vector $y$. It is, however, straightforward to observe that the result in the previous theorem can be shown also for any \emph{sequential} best-response dynamics starting from $y^1 = (1,\ldots,1)$, in which agents deviate one-by-one. We omit a formal adjustment of the proof.

In many game-theoretic scenarios, concurrent deviation gives rise to oscillation. The next example shows that fire sale games are no exception to this rule. We work out the example below for linear prices $p_j$ and $\alpha = 1$. It is easy to see that similar examples exist for $\alpha = 1$ with monotonic prices $p_j$, or $\alpha \in (0,1)$ with convex prices.
\begin{example}
    \label{ex:cycle1}
    Consider a game with two agents, a single asset, and linear price impact, where $p_1^0=1$, and $\alpha = 1$. The external assets and liabilities are such that $a_1^I = a_2^I= 1$ and $\ell_1 = \ell_2 = 5/4$. Moreover, both players hold half of the security, i.e., $x_{11}=x_{21}=1/2$. We assume $\lambda = 6$. If both agents play $y_1 = y_2 = 1$, then the leverage is $1.5/(1/4) = 6 = \lambda$. If Agent 1 plays $y_1 = 0$, then Agent 2 is illiquid and, thus, her best response is $y_2 = 0$. Now suppose we start in state $y^1 = (1,0)$ and let the agents deviate simultaneously. Then the resulting state oscillates between $y^{2i} = (0,1)$ and $y^{2i+1} = (1,0)$. \hfill $\lhd$
\end{example}

If in the example the two agents deviate sequentially, then after one step an equilibrium is reached. More generally, we show below that for all fire sale games with two agents and monotonic best responses, there can be no cycle in sequential best-response dynamics, no matter from which initial state the dynamics starts. 

\begin{restatable}{prop}{propseven}
    Consider a fire sale game with two agents and assume that the best-response function is monotonic. Then every sequential best-response dynamics is acyclic.
\end{restatable}

\begin{proof}
    First, consider the case that in some round $t$ we have $\Phi(y^{(t)}) \ge y^{(t)}$, i.e., the best response is at least the current strategy for both agents. By monotonicity of the best-response function, the agents will only keep increasing their strategies, which makes a cycle impossible. A similar argument shows the result when in some round $t$ we have $\Phi(y^{(t)}) \le y^{(t)}$.
    
    Now suppose that in some round $t$, the best response for Agent 1 is at most $y_1^{(t)}$ and for Agent 2 it is at least $y_2^{(t)}$. Suppose Agent 1 moves in round $t+1$. By monotonicity, this decreases the best response for Agent 2. After round $t+1$, the best response of Agent 2 is either at most or at least $y_2^{(t+1)} = y_2^{(t)}$. Agent 1 is playing her exact best response (i.e., at most \emph{and} at least the best response). Hence, one of the previously considered cases applies. A symmetric argument applies if Agent 2 moves first.
\end{proof}

We conjecture that there are fire sale games with two agents, for $\alpha = 1$ and monotonic prices or $\alpha \in (0,1)$ and convex prices, in which the best-response function is not continuous from below. Then given a state where all agents want to increase their strategies, sequential best-response dynamics might not reach an equilibrium in the limit. In contrast, suppose that the agents are in a state where they all want to decrease their strategies. Then the proof of \Cref{thm:br-convergence} can be applied to show that, for $\alpha = 1$ and monotonic prices or $\alpha \in (0,1)$ and convex prices, the limit of the best-response dynamics is indeed an equilibrium, for any number of agents. 

For three or more agents, if in the initial state there are agents above and below their best response, sequential best-response dynamics can again exhibit cyclical behavior. Below we discuss an example with three agents, $\alpha = 1$, and convex price impact. It is minimal in the sense that for two agents, sequential best-response dynamics are always acyclic in this case.

\begin{example}\label{ex:best-response-cycle-even}
Consider a fire sale game with
three agents and three assets where $\alpha=1$ and $\lambda=6.2$.
Furthermore, let $a_{i}^{I}=100$ and $l_{i}=90$ for all $i\in\{1,2,3\}$
and let the asset holdings be as follows:

\begin{center}%
\begin{tabular}{ccc}
$x_{11}=0.8$  & $ x_{12}=0.2$ & \tabularnewline
& $x_{22}=0.8$  & $x_{23}=0.2$ \tabularnewline
$x_{31}=0.2$  &  & $x_{33}=0.8$ \tabularnewline
\end{tabular}
\end{center}

Let $p_{j}(y)=10\cdot\left(\sum_{i}x_{ij}y_{i}\right)^{2}$ be the
price function for all assets $j\in\{1,2,3\}$. Note that price
impact is convex. A best-response cycle is given in the following
table:
\begin{center}
\begin{tabular}{rrr|rrr}
$y_{1}$  & $y_{2}$  & $y_{3}$  & $u_{1}(y)$ & $u_{2}(y)$ & $u_{3}(y)$\tabularnewline
\hline 
1  & 1  & 0 & $-\infty$ & 18.08 & 11.6 \tabularnewline
0  & 1  & 0 & 11.28 & $-\infty$ & 10.32 \tabularnewline
0  & 1  & 1 & 11.6 & $-\infty$ & 18.08 \tabularnewline
0  & 0  & 1 & 10.32 & 11.28 & $-\infty$ \tabularnewline
1  & 0  & 1 & 18.08 & 11.6  & $-\infty$\tabularnewline
1  & 0  & 0 & $-\infty$ & 10.32 & 11.28 \tabularnewline
1  & 1  & 0 & $-\infty$ & 18.08 & 11.6 \tabularnewline
 & \vdots & & & \vdots & \tabularnewline
\end{tabular}
\end{center}
Every agent maximizes her utility by selling as little as possible. Due to convex price impact, each agent is particularly affected by price devaluation of an asset of which she holds a large share. These properties lead to the following cyclic behavior:
Consider the first entry in the table. Agent 1 is illiquid due to the sales of Agent 3, with Agent 3 overselling. In the next step, Agent 1 liquidates all assets, causing illiquidity of Agent 2. Only afterwards does Agent 3 deviate to her best response and sell nothing. This again leaves Agent 1 overselling and Agent 2 illiquid.

In this game, the pointwise maximal equilibrium is obtained when no assets are sold, i.e., $y^1$. Thus, all strategy profiles played in the best-response cycle are point-wise below the maximal equilibrium. 

Also, note that the best-response cycle in this example is robust to different choices of $\alpha$ (e.g., using $\alpha = 0.5$). \hfill $\lhd$
\end{example}

\subsection{Concave Price Impact}
\label{sec:concavePrices}

\Cref{thm:lattice} shows that the equilibria of a fire
sale game have a very desirable shape if $\alpha=1$ or price impact
is convex. In this subsection, we discuss the converse case. When $\alpha<1$ and price impact is not convex, we show
that the set of equilibria need not exhibit a lattice structure.

We first show that the best-response function need not be monotonic
in this case. To develop some intuition for this, consider $\alpha<1$
and a concave and piecewise-linear price impact function for a given
asset $j$.
This means that there are \emph{tipping points} such that,
whenever the total sales in an asset reach such a point, the price
decreases according to a linear function with greater slope. By applying \Cref{prop:p-convex-e-convex} to each linear (and thus trivially convex) segment,
we receive that the equity function is piece-wise convex (but
not necessarily globally convex). More in detail, for every linear
section we receive high initial losses due to sales and a “bank run”
effect as the sales progress further. Whenever the overall sales pass
a tipping point, the equity function enters the next convex section.
Therefore, for a fixed asset and a sufficiently small $\alpha$, each
agent maximizes their utility by either (1) selling as little as
possible, (2) selling everything or (3) selling such that the overall
sales equal a tipping point.
In case (3), the agent's strategy depends directly on the other agents.
The following example shows that this effect can lead to non-monotonicity
of the best-response function.

\begin{restatable}{prop}{propeight}
There exists a fire sale game with a single asset, $\alpha\in(0,1)$,
and concave price impact such that the best-response function $\Phi$
is not monotonic.
\end{restatable}

\begin{proof}
Consider the following example with one asset; omit the $j$ index
for assets. Agent $1$ is defined by $a_{1}^{I}=10,l_{1}=9,x_{1}=0.5$,
where all other agents together hold half of the asset. Furthermore,
the leverage constraint is given by $\lambda=8.3$ and the price function
is defined as 
\begin{align*}
p(y_{1},y_{-1})=\begin{cases}
7+0.1\sum_{i}x_{i}y_{i} & \text{, if }\sum_{i}x_{i}y_{i}\geq0.5\\
5.05+4\sum_{i}x_{i}y_{i} & \text{, otherwise.}
\end{cases}
\end{align*}
Assume that the other agents sell according to a strategy profile
$\hat{y}$ where $\sum_{i=2}^{n}x_{i}\hat{y}_{i}=0.4$. Agent $1$
then fulfills the leverage constraint when selling at least according
to $\hat{y}_{1}=0.7899$, leading to $u_{1}(\hat{y}_{i},\hat{y}_{-i})\approx1.541$.
The best response of agent $1$ is given by $\hat{y}_{1}^{*}=(0.5-\sum_{i=2}^{n}x_{i}\hat{y}_{i})/x_{1}=0.2$
where exactly the tipping point of $p$ is reached. This strategy
satisfies the leverage constraint and yields a utility of $u_{1}(0.2,\hat{y}_{-1})\approx1.543$.

Now, assume the other agents increase the sales to $y_{-1}\leq\hat{y}_{-1}$
where $\sum_{i=2}^{n}x_{i}y_{i}=0.3$. Then, agent $1$ is forced
to sell at least according to $y_{1}=0.781$ with $u_{1}(y_{1},y_{-1})\approx1.538$.
Again, the best response of Agent 1 is to exactly match the tipping
point of $p$. Therefore, the maximal utility of agent $1$ is attained
at $y_{1}^{*}=(0.5-\sum_{i=2}^{n}x_{i}y_{i})/x_{1}=0.4>0.2$ with
$u_{1}(y_{1}^{*},y_{-1})\approx1.539$. 
\end{proof}

The proposition implies that our proof for \Cref{thm:lattice} cannot be applied in the case of non-convex price impact and $\alpha\in(0,1)$.
Indeed, an equilibrium need not even exist in this case.

\begin{restatable}{prop}{propnine}
There exists a fire sale game with $\alpha\in(0,1)$
and concave price impact where no equilibrium exists. 
\end{restatable}

\begin{proof}
Consider a game with two agents and three assets. Let agent
1 be defined by $a_{1}^{I}=750,l_{1}=900$ and $x_{11}=1,x_{12}=0.8,x_{13}=0$.
Agent 2 is defined by $a_{2}^{I}=11850,l_{2}=11000$ and $x_{21}=0,x_{22}=0.2,x_{23}=1$.
Furthermore, let $\alpha=0.000006$ and $\lambda=2$ and let the price
functions be given by 
\begin{align*}
p_{1}(y) & =\begin{cases}
9.99+0.01\sum x_{i1}y_{i} & \text{, if }\sum x_{i1}y_{i}\geq0.1\\
\frac{9.991}{0.1}\sum x_{i1}y_{i} & \text{, otherwise}
\end{cases}\\
p_{2}(y) & =\begin{cases}
999.99+0.01\sum x_{i2}y_{i} & \text{, if }\sum x_{i2}y_{i}\geq0.3\\
\frac{999.993}{0.3}\sum x_{i2}y_{i} & \text{, otherwise}
\end{cases}\\
p_{3}(y) & =10000\sum x_{i3}y_{i}.
\end{align*}
The game has no equilibrium if for every strategy profile $y$ at
least one agent profits by unilaterally deviating. First, consider Agent 2 who holds
a fraction of Assets 2 and 3. As Asset 3 has much higher value
than Asset 2 and $x_{23}>x_{22}$, the utility of Agent 2 is dominated
by the sales of Asset 3. Hence, Agent 2 will optimize her strategy with respect to Asset 3. Consequently, she either minimizes her sales or liquidates all assets due to a “bank run” effect. Therefore, we restrict our discussion to these two strategies. 
Strategy $y_{2}^{1}$ is feasible for all $y_{1}\in[0.031,1]$ and maximizes the equity for all $y_{1}\in[0.125,1]$.
Hence, $y_{2}^{1}$ is the best response for $y_{1}\in[0.125,1]$.
In contrast, strategy $y_{2}^{0}$ is always feasible by definition
and is therefore the best response for $y_{1}\in[0,0.125)$.

Since Agent 2 will always deviate to one of the two strategies, it
is sufficient to consider the corresponding best responses of Agent
1. First, assume Agent 2 plays $y_2^1$,
in which case the feasible strategies of Agent 1 are given
by $y_{1}\in[0,0.7037]$. The best
response is $y_{1}^{0}$ with utility $u_{1}(y_{1}^{0},y_{2}^{1})=659.9983$.
Now, if Agent 2 sells everything according to $y_{2}^{0}$, the best
response of Agent 1 is $y_{1}^{*}=0.375$ with utility $u_{1}(y_{1}^{*},y_{2}^{0})=659,996$,
which exactly corresponds to the tipping point of $p_{2}$.

In conclusion, for any starting point, the agents deviate cyclically
where Agent 1 plays strategies 0 and 1 and Agent 2 sells according
to 0 and 0.375. In particular, no equilibrium exists. 
\end{proof}

In the example above, the non-monotonicity of the best-response function of Agent 1 is exploited. The agent either liquidates all assets or sells in such a way that the overall sales correspond to the tipping point of Asset 2. 
Suppose Agent 2 sells nothing, then strategy $y_1^{t_1}=0.3$ corresponds to the tipping point. In this case, Agent 1 maximizes her utility by selling everything.
However, if Agent 2 sells everything, the strategy corresponding to the tipping point shifts to $y_1^{t_0}=0.375$ and becomes more profitable than $y_1^0$. 
In conclusion, Agent 1's best-response function is monotonic, while Agent 2 wants to reduce her sales when Agent 1 liquidates more. Thus, the agents deviate in a circular fashion.

\section{Quality of Equilibria\label{sec:poa}}

In this section, we discuss the quality of Nash equilibria
of fire sale games. We show that in the case of
monotonic equity (for instance, $\alpha=1$), every Pareto-optimal strategy profile forms a Nash equilibrium. Moreover, it even represents a \emph{strong} equilibrium. 
\global\long\def\yopt{y^{\text{opt}}}%
\global\long\def\sw{\operatorname{sw}}%

We consider utilitarian welfare as the social welfare measure. Consider
a strategy profile $y$ and assume w.l.o.g.\ that every insolvent or illiquid
agent $i$ plays $y_{i}=0$ (otherwise the social welfare is $-\infty$). Then we have
\begin{align*}
\sw(y) & =\sum_{i}u_{i}(y)=\sum_i e_{i}(y)
 =\sum_i a_{i}^{I}-\sum_{i}l_{i}+\alpha\sum_{i,j}x_{ij}p_{j}(y) + (1-\alpha)\sum_{i,j}x_{ij}\left(y_i p_{j}(y)+(1-y_i)p_{j}^{0}\right)
 .
\end{align*}

Note that, if $\alpha<1$, it can be optimal in terms of social welfare to have
a large number of asset sales because price impact is partially avoided
by high sales. For $\alpha=0$, it is easy to see that $y=y^{0}$ is a dominant-strategy
equilibrium and always social-welfare optimal. We consider this effect largely an artifact
of the model for unrealistic values of $\alpha$. For $\alpha=1$,
the expression simplifies to
\begin{align*}
\sw(y) & =\sum_{i}a_{i}^{I}-\sum_{i}l_{i}+\sum_j p_{j}(y),
\end{align*}
so that social welfare is maximized by maximizing the overall value of the market, subject to satisfying all leverage constraints or choosing $y_i=0$ for illiquid banks.
The following theorem shows that in fire sale games with monotonic equities, every profile with optimal social welfare is also very stable.
More precisely, we show that all \emph{Pareto optima} are strong equilibria.

\begin{restatable}{thm}{theoremXX}
\label{thm:pos1}
Suppose the equity $e_i(y_i, \ym)$ of each player $i$ is monotonic in $y_i$. Then every Pareto optimum is a strong equilibrium.
In particular, every profile with optimal social welfare is a strong equilibrium. 
\end{restatable}
\begin{proof}
    Let $y$ be a Pareto optimum and assume towards a contradiction that $y$ is not a strong equilibrium. Then there is a coalition $C\subseteq N$ such that all agents in $C$ strictly benefit by deviating to some strategy profile $y'=(y'_C, y_{-C})$. Let $C^+=\{i \in C \mid y'_i > y_i\}$ be the set of agents that reduce their sales. $C^+$ is non-empty; otherwise, $e_i(y') \le e_i(y)\;\forall i$ by monotonicity of the equity (by assumption).
    
    Let $y'' = (y'_{C^+}, y_{-C^+})$, i.e., only the agents in $C^+$ deviate. Note that $y'' \ge y',\, y$. We show that $y''$ is a strict Pareto improvement over $y$, contradicting $y$'s optimality. We perform a case distinction over whether an agent $i$ is a member of the set $C^+$.
    
    First, consider $i\in C^+$. We show that $u_i(y'') > u_i(y)$. Since by assumption, agent $i$ had an incentive to change her strategy to $y_i' > 0$ under $\ym'$, it must be the case that she satisfies the leverage constraint at $y'$ (otherwise, she would derive utility $-\infty$). Therefore, $\lev_i(y'') = \lev_i(y_i', \ym'') \le \lev_i(y_i', \ym')\le \lambda$, where the equality is by definition and the first inequality is by monotonicity of $\lev_i$ (\Cref{prop:leverage_inc}).
    Now we have
    \begin{align*}
        u_i(y'') &= e_i(y'') = e_i(y_i', \ym'') \ge e_i(y_i', \ym')
        = u_i(y') > u_i(y).
    \end{align*}
    Here, the first inequality is by monotonicity of the equity in $\ym$ and the second is by assumption.
    
    Now consider $i \notin C^+$. Then we have $u_i(y'') = u_i(y_i, \ym'') \ge u_i(y_i, \ym)$ because $\ym'' \ge \ym$.
    More in detail, going from $\ym$ to $\ym''$ while keeping $y_i$ the same can only move agent $i$ from insolvency to solvency (but not the other way round, by monotonicity of $e_i$ in $\ym$) and from illiquidity to liquidity (but not the other way round, by monotonicity of $\lev_i$ in $\ym$, see \Cref{prop:leverage_inc}), and $e_i$ increases.
\end{proof}

For games with post-sale prices and $\alpha = 1$, \Cref{thm:lattice} shows that the set of Nash equilibria forms a complete lattice. Moreover, due to \Cref{prop:monotonic_post_br-fkt} and the definitions of equity and utility, the supremum of the Nash equilibria has the highest social welfare among all Nash equilibria. \Cref{thm:pos1} shows that in these games, every profile with optimal social welfare is a strong equilibrium. Since every strong equilibrium is a Nash equilibrium, we obtain the following corollary.

\begin{cor}\label{cor:PoS}
For fire sale games with $\alpha = 1$, the supremum of all Nash equilibria is a strong equilibrium and represents a profile with optimal social welfare. 
\end{cor}

The previous theorem depends on monotonicity of the equity. In the following theorem, we observe that a monotonic best-response function $\Phi$ alone is not sufficient to show the result. More precisely, we show that in games with $\alpha \in (0,1)$ and convex prices, the social optimum does not necessarily form an equilibrium. Moreover, strong equilibria may not even exist.

\begin{restatable}{thm}{theoremsix}
    There is a fire sale game with $\alpha \in (0,1)$ and convex prices without strong equilibria. The social optimum does not form a Nash equilibrium.
\end{restatable}

\begin{proof}
Consider the following fire sale game with two agents and two assets. The implementation shortfall is $\alpha=0.5$ and the leverage constraint is defined as $\lambda=1.48971$. The price function $p_1(y)=(\sum_{i\in N} x_{i1}y_{i})^{2}$
is convex with $p_1^{0}=1$ while $p_2$ is linear with $p_2^0=0.5$. Let Agent 1 be defined by $a_1^I=2.2$, $l_1=1$ and the share $x_{11}=0.5$. Furthermore, Agent 2 is given by $a_2^I\approx 2.35$, $l_2=1$ and $x_{21}=0.5, x_{22}=1$.
The social optimum of ca.~$3.55$ is achieved with the strategy
profile $\yopt=(0.25,1)$, where both agents satisfy the leverage constraint
with minimal sales. For this strategy profile, the utilities of the agents are $u_1(\yopt)\approx 1.51$ and $u_2(\yopt)\approx 2.04$. However, Agent 1 can strictly improve her utility by unilaterally deviating due to the bank-run effect. When liquidating all assets, the utility of Agent 1 is given by $u_{1}(0,\yopt_{-1})\approx1.51$. As a consequence, Agent 2 has to increase her sales to  $y_2\approx 0.93$ resulting in the unique equilibrium $y=(y_1^0,y_2)$.
Since the utilities are given by $u_1(y)\approx 1.5$ and $u_2(y)\approx 1.92$, the agents have an incentive to coalitionally deviate to the social optimum.
\end{proof}

\Cref{cor:PoS} states that, for fire sale games with post-sale prices, the supremum of the Nash equilibria is an optimum in terms of social welfare (and, hence, a Pareto optimum) and forms a strong equilibrium. The next theorem extends this result and shows that this profile is the \emph{unique} strong equilibrium and, hence, the unique Pareto optimum and the unique optimum in terms of social welfare.

\begin{restatable}{thm}{theoremseven}\label{thm:sNE}
For fire sale games with $\alpha=1$  
there is a unique strong Nash equilibrium.
\end{restatable}

\begin{proof}
    The supremum of all Nash equilibria $y^{opt}$ forms a strong equilibrium that is optimal in terms of social welfare (\Cref{cor:PoS}). We will show that $y^{opt}$ is the unique strong equilibrium. 
    
    For contradiction, suppose there is some other strong equilibrium $y \neq y^{opt}$. Then $y$ must also be a Nash equilibrium, and by the lattice structure of Nash equilibria (\Cref{thm:lattice}) we know $y_i \leq y^{opt}_i$ for all $i \in N$. Let $C = \{i \in N \mid y_i < y_i^{opt}\}$ be the set of agents that play different strategies in $y$ and $y^{opt}$. Consider the coalitional deviation from agents in $C$ from $y$ to $y^{opt}$. We denote the resulting profile by $(y_C^{opt}, y_{-C})$. For every affected security $j$ with $\sum_{i \in C} x_{ij} > 0$ we obtain
    $$\sum_{i \in N} y_i x_{ij} < \sum_{i \in C} y^{opt}_i x_{ij} + \sum_{i \notin C} y_i x_{ij}\enspace.$$
    Hence, for the price of that security, 
    $$p_j(y) = p_j\left(\sum_{i \in N} y_i x_{ij}\right) < p_j\left(\sum_{i \in C} y^{opt}_i x_{ij} + \sum_{i \notin C} y_i x_{ij}\right) = p_j(y_C^{opt}, y_{-C}),$$
    since prices of all securities are increasing. Thus, by the definition of utility and equity, every agent in $C$ strictly profits. This contradicts the fact that $y$ is a strong equilibrium.
\end{proof}

In contrast to \Cref{thm:sNE} for strong equilibria, there can exist multiple Nash equilibria, some with devastating social welfare. The next example represents a game with $\alpha = 1$ and linear price impact, in which every agent satisfies the leverage constraint when liquidating no assets. However, when all agents sell all securities, this also represents a Nash equilibrium. Hence, while the optimal profile is a Nash equilibrium, the worst Nash equilibrium is a profile in which all agents are illiquid.

\begin{example}
	Consider a game with $n$ agents where all agents are characterized by the same properties. More in detail, let $a_i^I>0$ be arbitrary, but equal across $i \in N$ and let $x_{i1}=\frac{1}{n}$ for all $i$ and the single security. Further, let the $p_1^0=\frac n2$, and let $l_i = a_i^I$ for all $i$. We define $\lambda$ as the leverage of every agent before fire sales, i.e.,
    \begin{align*}
	\lev_i(y^1) = \frac{a_i^I +\frac{1}{n} \frac n2}{a_i^I - a_i^I + \frac 1 n  \frac n2} = 2a_i^I +1=: \lambda \enspace,
	\end{align*}
    where $y^1=(1,1,\dots,1)$.
	By definition of $l_i$ and $\lambda$, no agent is insolvent or illiquid. Since every agent maximizes her utility by playing according to $y^1$, the strategy profile is optimal and yields a social welfare of
    \begin{align*}
	   \sw(y^1) &= n \cdot\left(a_i^I - l_i + x_{i1}p_1^0\right) = n\cdot\left(a_i^I - a_i^I +\frac{1}{n} \frac n2\right) = \frac n2 \enspace .
	\end{align*}

	Now, let all agents sell according to $y^0=(0,0,\dots,0)$. No agent $i$ can unilaterally profit from a strategy $y_i \in (0,1]$, since, for every $i \in N$,
    \begin{align*}
	   \lev_i(y_i,y_{-i}^0) &= \frac{a_i^I + y_i x_{i1} p_1(y_i,y_{-i}^0)}{a_i^I -l_i + x_{i1} p_1(y_i,y_{-i}^0)} = \frac{a_i^I + y_i \frac{1}{n} \frac n2\cdot(y_i \frac 1n)}{a_i^I - a_i^I + \frac 1n \frac n2 \cdot( y_i \frac{1}{n})}\\
	   &= \frac{2a_i^I n}{y_i} + y_i > 2a_i^I + 1 = \lambda
	\end{align*}
	where the inequality holds for sufficiently large $n$. Thus, $y_i^0$ is the unique feasible strategy for every agent $i$ with respect to $y_{-i}^0$. Hence, $y^0$ forms a Nash equilibrium with social welfare
	\begin{align*}
	    \sw(y^0) = n \cdot \left(a_i^I - l_i + x_{i1} p_1(y^0)\right) = n \cdot \left(a_i^I - a_i^I + \frac 1n p_1^0\cdot(0)\right) = 0 \enspace . 
	\end{align*} \hfill $\lhd$
\end{example}

\section{Convergence Speed of Dynamics\label{sec:convergence-dynamics}}

In this section, we study two dynamics by which agents may reach an
equilibrium: the standard best-response dynamics and a simplified dynamics where
agents neglect their own price impact. We examine the convergence of these dynamics towards an equilibrium. We focus on $\alpha=1$ and linear price impact.

\subsection{Best-Response Dynamics}
\Cref{thm:br-convergence} shows that (in particular)
for the case $\alpha=1$, best-response dynamics starting at $y^{1}$
always converge to the point-wise maximal equilibrium (where asset
sales are collectively minimized). If the sequence
proceeds for a finite number of steps, we show that it reaches an \emph{approximate
equilibrium} quickly, assuming the numeric values of the game are reasonably large.

The traditional concept of an approximate Nash equilibrium is not
appropriate for fire sale games because even a small violation of the leverage
constraints leads to an infinite decrease in
utility. We therefore define an \textit{approximate equilibrium} in a fire
sale game as a strategy profile where agents can only improve their
equity by a small amount and where the leverage constraints are approximately
satisfied.

\begin{defn}
Let $y$ be a strategy profile in a fire sale game and let $\eps>0$.
Then $y$ is an \emph{$\eps$-approximate equilibrium} iff the following
hold for every agent $i$:
\begin{enumerate}
\item If $i$ is liquid, then for any $y_{i}'$ such that $i$ is liquid
for $(y_{i}',\,y_{-i})$ we have $e_{i}(y)\ge e_{i}(y_{i}',\,y_{-i})-\eps$.
\item If $i$ is liquid, then $\lev_{i}(y)\le\lambda+\eps$.
\item If $i$ is insolvent or illiquid, then $y_{i}\le\eps$. \hfill $\lhd$
\end{enumerate}
\end{defn}

An approximate equilibrium might not be close to any exact equilibrium in terms of norm distance in strategy space, a common property of approximate equilibrium concepts (see, e.g., \citet{etessami2010complexity} for a discussion of this phenomenon in the context of approximate Nash equilibria).

We show that best-response dynamics reach an $\eps$-approximate equilibrium in pseudo-polynomial time.
\begin{restatable}{thm}{theoremthree}
Consider a fire sale game with $\alpha=1$ and linear
price impact.
Let $x_{\max}$ be the maximum over all values $x$ and $x^{-1}$, where $x$ is a numeric value contained in the input.
Let $\eps>0$.
Then the best-response dynamics, after $n/\eps$ steps, reaches a
point $y^*$ such that $\norm{\Phi(y^*)-y^*}_{\infty}\le\eps$ and $y^*$
is a $(\poly(x_{\max}) \cdot \eps)$-approximate equilibrium.
\end{restatable}
\begin{proof}
As long as $\norm{\Phi(y^*)-y^*} > \eps$, trivially, some $y_i^*$ decreases by at least $\eps$ in every step. As $y^*$ is bounded below by $(\fromto 0 0)$, there can be at most $n/\eps$ such steps.

We now show that $y^*$ is an $\eps$-approximate equilibrium.
First, consider an insolvent or illiquid agent $i$. Then $\Phi_i(y^*)=0$ and thus, since $\norm {\Phi_i(y^*) - y_i^*} \le \eps$, we have $y_i^* \le \eps$ as required.

Let now $i$ be liquid. Let $y' = \Phi(y^*)$ and assume that $y_i' < 1$ (otherwise, $y^*$ even satisfies the requirements for an \emph{exact} equilibrium at $i$). Then by choice of $y_i'$ we have $\lev_i(y_i', \ym^*) = \lambda$. We bound the difference $\lev_i(y_i^*, \ym^*) - \lev_i(y_i', \ym^*)$.
To do this note that, as the sequence is descending, $y_i' \le y_i^* \le y_i' + \eps$. Consider the derivative
\[
    \frac {\de\, \lev_i(y)} {\de y_i} =
    \left(\frac {\frac {\de a_i} {\de y_i} e_i - \frac {\de e_i} {\de y_i} a_i} {e_i^2}\right)(y)
    =: \frac {N(y)} {e_i^2(y)}
    .
\]
As $e_i$ is increasing in $y_i$, for any $y_i \in [y_i', y_i^*]$ we have $e_i(y_i, \ym^*) \ge e_i(y_i', \ym^*) = a_i(y_i', \ym^*)/\lambda \ge a_i^I / \lambda \ge x_{\max}^{-2}$, where the first equality is because $\lev_i(y_i', \ym^*)=\lambda$.

It is easy to see that $N(y)$ is a polynomial (of degree 2) in $y$ where values of the coefficients, but not the structure of the function, depend on the input. Since further $y \in [0,1]^n$, we have $N(y) \le \poly(x_{\max})$ and thus
\[
    \frac {\de\, \lev_i(y)} {\de y_i} \le \poly(x_{\max})
    .
\]
By integration, we have $\lev_i(y^*) \le \lambda + \poly(x_{\max}) \cdot \eps$.

Finally, we show that no liquid agent can improve her equity by more than $\poly(x_{\max})\cdot\eps$.
This follows using the same technique as above because we can bound the derivative $\frac {\de e_i(y)} {y_i}$ to receive $e_i(y^*) \ge e_i(y_i', \ym^*) - \poly(x_{\max})\cdot\eps$ and $y_i'$ maximizes $e_i(\blank, \ym^*)$ by definition.
\end{proof}

\subsection{Simplified Best-Response Dynamics}
Computing a best response is relatively complex as the agent
needs to take into account the price impact of the very sales she is about to decide on. For some agents, this may be unrealistic.
Hence, we consider a simplified dynamics, where agents neglect
their own price impact. Similar ``best-response'' dynamics have been considered by \citet{cont2017fire}.
Here, liquidations proceed in several rounds and agents consider prices
as fixed during each round.

We first define simplified versions of the different components of
each agent's wealth where the price impact of the agent's current
choice of strategy is excluded. These functions take one parameter
more compared to the full versions in \Cref{sec:Model} to
differentiate between the current choice and a previous choice by
the agent.

\global\long\def\yt{\tilde{y}}%
 
\global\long\def\att{\tilde{a}}%
 
\global\long\def\Dt{\tilde{\Delta}}%
 
\global\long\def\et{\tilde{e}}%
\global\long\def\levt{\tilde{\operatorname{lev}}}%
 
\begin{defn}
For $\yt_{i}\in[0,1]$ and $y\in[0,1]^{N}$, define the \emph{simplified}
assets, revenue, equity, and leverage as the respective terms from \Cref{sec:Model} where, however, the price impact is calculated
based on $y$, but agent $i$'s sales are calculated according to
$\yt_{i}$. More in detail, let 

\begin{alignat*}{2}
\att_{i}(\yt_{i},y) & :=a_{i}^{I}+\yt_{i}\sum_{j}x_{ij}p_{j}(y) &  & =a_{i}^{I}+\yt_{i}V_{i}(y)\\
\Dt_{i}(\yt_{i},y) & :=\sum_{j}x_{ij}(1-\yt_{i})p_{j}(y) &  & =(1-\yt_{i})V_{i}(y)\\
\et_{i}(\yt_{i},y) & :=\att_{i}+\Dt_{i}-l_{i} &  & =a_{i}^{I}-l_{i}+V_{i}(y)\\
\levt(\yt_{i},y) & :=\frac{\att_{i}(\yt,y)}{\et_{i}(y)}\\
\text{where}\quad V_{i}(y) & :=\sum_{j}x_{ij}p_{j}(y).
\end{alignat*}
\hfill $\lhd$
\end{defn}

Observe that $V_{i}(y)$ is the value of $i$'s liquid asset holdings
if $i$ sells nothing and price impact is given by $y$. Further observe
that $\et_{i}(\yt_{i},y)$ is in fact independent of $\yt_{i}$ and (thus) we have $\et_i(\yt_i, y) = e_i(y)$ for all $\yt_i$. This
is because, under the assumption of the simplified best-response dynamics, sales according to $\yt_{i}$
do not generate any additional price impact and thus they transform
assets (valued at market price) into risk-free assets at a rate
of 1; these terms cancel out in the calculation of the equity. We
extend our model by making the realistic assumption that agents still
aim to sell as little as possible (i.e., maximize $\yt_{i}$) subject
to meeting their leverage constraint. In this case, the best response
of $i$ according to the simplified dynamics is easily calculated:
\begin{defn}
For $y\in[0,1]^{N}$ let
\[
g_{i}(y):=\lambda-\frac{\lambda l_{i}-(\lambda-1)a_{i}^{I}}{V_{i}(y)}\enspace.
\]
The \emph{simplified best-response function} $\Psi:[0,1]^{N} \to[0,1]^{N}$ is given by
\begin{align*}
\Psi_{i}(y) & :=\begin{cases}
\min(1,\max(0,g_{i}(y))) & \text{if }e_i(y) > 0\\
0 & \text{if }e_i(y)\le 0.
\end{cases}
\end{align*}
\hfill $\lhd$
\end{defn}

\begin{restatable}{lem}{lemmatwo}
\label{lem:simplified-br-correct}The value $\yt_{i}:=\Psi_{i}(y)$
is the maximal $\yt_{i}$ such that $\et_{i}(\yt_{i},y)>0$ and $\levt_{i}(\yt_{i},y)\le\lambda$,
if such a $\yt_{i}$ exists. Otherwise, $\Psi_{i}(y)=0$.
\end{restatable}
\begin{proof}
Recall that $\et(\yt_i, y) = e_i(y)\;\forall y$.
If $e_i(y)\le 0$, then
the statement is trivial.
So assume $e_i(y) > 0$.
It is easy to see that $g_i(y)$ is such that
\begin{equation}
    0 = \att_i(g_i(y), y) - \lambda \et_i(g_i(y), y) = \att_i(g_i(y), y) - \lambda e_i(y)
.
\label{eq:gi-prop}
\end{equation}
If $e_i(y) \le 0$, then we must have $g_i(y) \le 0$ and thus $\Psi_i(y)=0$ by definition.
If $e_i(y) > 0$, then Equality~\eqref{eq:gi-prop} implies that (a) $g_i(y) \ge 0$ and (b) $\levt_i(g_i(y), y) = \lambda$. Because $\levt_i(\yt_i, y)$ is strictly monotonic in $\yt_i$ (because of linear price impact), this implies that $\Psi_i(y) = \min(1,g_i(y))$ has the properties of $\yt_i$ as needed.
\end{proof}

When agents act according to the simplified best-response function
$\Psi$, they ignore their own price impact. However, the
price impact resulting from these sales enters in the next round,
where it affects all agents, including the ones who increased their
sales in this round. Thus, no information is lost.
As the following theorem shows, this implies that simplified best-response
dynamics also converge to the maximal equilibrium, just like the
best-response dynamics.

\begin{restatable}{thm}{theoremfour}
\label{thm:simpl-br-converges}Consider a fire sale game with $\alpha=1$ and linear price impact. Then the following hold:
\begin{enumerate}
\item $\Psi$ is monotonic and continuous.
\item $\Phi$ and $\Psi$ have the same fixed points.
\item Let $\left(\yt^{t}\right)$ be a sequence of strategy profiles defined
by $\yt^{0}=(\fromto 11)$ and $\yt^{t+1}=\Psi(y^{t})$. Then $\left(\yt^{t}\right)$
is point-wise monotonically decreasing and converges
to the point-wise maximal equilibrium of the fire sale game.
\end{enumerate}
\end{restatable}

\begin{proof}
1. Monotonicity follows from the definition because the condition
$e_i(y)$ is increasing in $y$ (i.e., it can only switch from false to true as $y$ increases point-wise, but not vice versa)
and the function $g_{i}$ is obviously monotonic. Towards continuity,
first note that $g_{i}$ is obviously continuous. For continuity between
the two cases of the case distinction, it follows from the proof of \Cref{lem:simplified-br-correct} that, as
$e_{i}(y)\to0$, $g_{i}(y)$ converges to a value $\le 0$.

2. This follows from \Cref{lem:simplified-br-correct} and the fact that $\et_i(y_i, y)=e(y)$ and $\levt_i(y_i, y)=\lev_i(y)$.
More in detail, if $y$ is any strategy profile for which some agent $i$ is insolvent or illiquid at $y_{-i}$, then $\Phi_i(y)=0$ and by \Cref{lem:simplified-br-correct} also $\Psi_i(y)=0$. Thus, $\Phi_i(y)=y_i$ iff $y_i=0$ iff $\Psi_i(y)=y_i$.
If agent $i$ is liquid at $y_{-i}$, then $y_i = \Phi_i(y)$ iff $y_i$ is maximal such that $e_i(y_i, \ym) > 0$ and $\lev_i(y_i, \ym) \le \lambda$. Since $\lev_i(y_i, \ym) = \levt_i(y_i, y)$ and by \Cref{lem:simplified-br-correct}, this is equivalent to $y_i=\Psi_i(y)$.

3. By part~1 and the same argument as in \Cref{thm:br-convergence},
the sequence converges to the maximal fixed point of $\Psi$. By part~2, this is also the maximal fixed point of $\Phi$, i.e., the point-wise
maximal equilibrium.
\end{proof}

While $\Phi$ and $\Psi$ have the same fixed points (i.e., equilibria of the fire sale game) and both converge to the maximal equilibrium, the speed of convergence can be vastly different, as illustrated in the following subsection.

\subsection{Experiments and Diversification}

We now study what properties of the asset holdings matrix $x$ affect the stability of the financial system.
Specifically, we are interested in the effect of \emph{diversification}, i.e., to which degree each agent spreads her investments across multiple assets. Higher diversification reduces the effect of each asset on each agent (increasing stability), but also increases the number of channels by which one agent may affect another (decreasing stability).
We are interested in the convergence speed of both standard and simplified best-response dynamics.
To that end, we perform a computational experiment with the required code available online~\citep{code}:

We consider games with an equal number of agents and assets $n=m$.
For each agent $i$, we draw $a_i^I$ uniformly in $[80, 120]$ and $l_i$ in $[40, 60]$.
Price impact is linear, with $\alpha=1$, and we set $p_j^0 = 100$ for each asset $j$.
We choose the asset holdings based on a parameter $\tau \in [0,1]$, which measures diversification.
We set $x_{ij} = \tau \cdot 1/n + (1-\tau) \cdot \delta_{ij}$, where $\delta_{ij}$ is the Kronecker symbol.\footnote{For $\tau=0$, each agent only holds a single (different) asset with no agent interactions; for $\tau=1$, all agents hold all assets equally.}
Let $\lambda^1$ be the highest leverage of any agent at $y^1=(1,\dots,1)$.
We draw $\lambda$ from $[0.6\lambda^1, 0.99\lambda^1]$ to ensure that $y^1$ is not a Nash equilibrium.
We reject all instances with $\lambda \leq 1$.
We let both dynamics converge from $y^1$ until all strategy changes in a step become smaller than $10^{-5}$.
\Cref{fig:stepsize} depicts the average size of strategy changes over time, while \Cref{fig:convergence} displays the number of steps to convergence.
For \Cref{fig:convergence}~(right) we uniformly drew $a_i^I$ from $[0, 100]$, $l_i$ from $[40, 100]$, $p_j^0$ from $[50, 150]$, and $\lambda$ from $[0.9\lambda^1, 0.99\lambda^1]$, with all other parameters unchanged.
For each data point, we average over $10^6$ runs.

In \Cref{fig:stepsize}, we see that the step size over time does not decrease exponentially for both dynamics, which suggests that convergence requires pseudo-polynomially many steps before reaching an approximate equilibrium on average.
We also see that diversification has a significant effect on the best-response dynamics, since an agent is less exposed to other agents for low values and thus needs to make smaller corrections after the initial steps.
The simplified dynamics, on the other hand, does not exhibit this behavior, as an agent is not accurately measuring her own price impact.
\Cref{fig:convergence} shows that for our examples, the simplified dynamics converge more slowly, but for many values of diversification not by a large amount.
Interestingly, a “hump” in convergence times appears at a certain diversification value, reminiscent of similar results on systemic risk \citep{elliott2014financial}.
However, a change of parameters may lead to wildly different effects of diversification on the convergence of the system for the simplified dynamics.
Thus, to draw further conclusions about convergence times, there needs to be sufficient information on the asset holdings.

\begin{figure}
    \centering
    \includegraphics{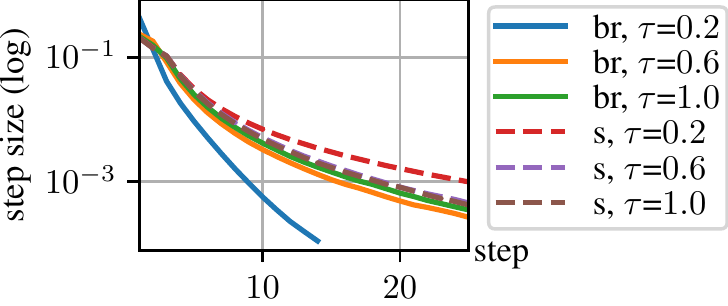}
    \caption{Comparison of the average step size over time of the best-response (br) and the simplified dynamics (s) for various values of diversification $\tau$ for games with 10 agents.}
    \label{fig:stepsize}
\end{figure}
\begin{figure}
    \centering
    \includegraphics{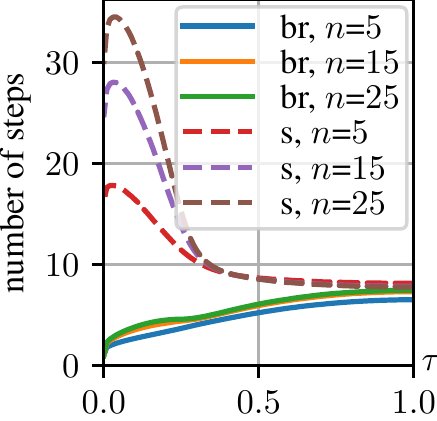}\includegraphics{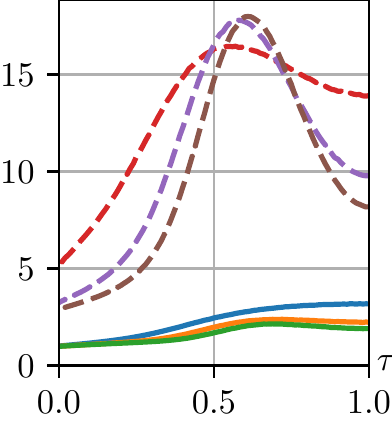}
    \caption{Comparison of the convergence time of the best-response (br) and the simplified dynamics (s) for varying number of banks and diversification $\tau$. The two graphs use different sets of parameters.}
    \label{fig:convergence}
\end{figure}

\section{Even and Non-Even Sales\label{sec:non-even}}

In the previous sections, we have focused on the case of proportional, or \textit{even}, sales, where each agent
can only choose a single number as her strategy and a fraction of her whole portfolio
is sold off according to this number. As a result, the relative
composition of each agent's portfolio does not change. In this section,
we discuss how this model differs from an extended case of (potentially)
non-even sales, where agents can choose each individual fraction of
each asset they want to keep. We will see that the non-even
case displays significant differences from the even case.

Hence, we extend the model in the following way:
The strategy of each agent $i$ now consists of a \emph{vector} (rather than a number) $\sq {y_{ij}} {j \in M} \in [0, 1]^M$, where $y_{ij}$ is the share of her holdings of asset $j$ that agent $i$ keeps. Consequently, we adjust the definitions for the total assets held and for the value of agent $i$'s assets and revenue from sales as follows:
\begin{align*}
x_{j}(y) & =\sum_{i}y_{ij}x_{ij}\\
a_{i}(y) & =a_{i}^{I}+\sum_{j}y_{ij}x_{ij}p_{j}(y)\\
\Delta_{i}(y) & =\sum_{j}(1-y_{ij})x_{ij}((1-\alpha)p_{j}^{0}+\alpha p_{j}(y))
\end{align*}
The other definitions, in particular those for the leverage $\lev_i$ and the equity $e_i$, stay the same up to these replacements.
We receive our original model in the special case where the $y_{ij}$ happen to be the same across $j$ for all $i$; thus, this is a true extension of our previous model.

We first highlight an important difference between the even-sales
and the non-even-sales case when $\alpha=1$: while in the case of
even sales, we know that higher sales by some agent will always lead to point-wise
higher sales by all other agents (i.e., the best-response function
$\Phi$ is monotonic, see \Cref{prop:monotonic_post_br-fkt}),
this is no longer true when non-even sales are allowed. In this case, the following theorem shows that
agents may respond to higher sales by shifting their own sales from
one asset to another.

Note that one cannot state this result in terms of a best-response \emph{function} because, for non-even sales, there may be complicated ties between several strategies that are all best responses. In general, there can be a surface of best responses.
This is why we do not define a best-response function here but instead frame our statement in terms of point-wise dominance of vectors.

\begin{prop}
There exists a fire sale game with non-even sales, $\alpha=1$, and linear price impact such that the following holds.
There is an agent $i$ and strategy profiles of the other agents $\ym \le \ym'$ such that, whenever $y_i$ is a best response to $\ym$ and $y_i'$ is a best response to $\ym'$, then $y_i \not\le y_i'$.
\end{prop}
\begin{proof}
	Consider a game with multiple agents and two assets.
	Agent $1$ is characterized by $a_1^I = 100$, $l_1 = 70$, $x_{11}=0.5$, and $x_{12}=0.5$ whereas the other agents together hold half of each asset, i.e., $\sum_{i\neq 1}x_{i1}=0.5$ and $\sum_{i\neq 1}x_{i2}=0.5$.
	Let $p_1^0=99$, $p_2^0=100$, and $\lambda = 1.5$.
	In reaction to the strategy $y_{-i}^1 = (\fromto 1 1)$ of the other agents,
	Agent $1$ maximizes her utility by selling according to $y_{11}=0.925998$ and $y_{12}=0.925997$ yielding the utility $u_1(y)=125.818$. 
	Now assume that the other agents increase their sales of Asset 2 in such a way that still $y_{i1}=1\;\forall i\neq 1$, but the strategies $y_{i2}$ are reduced so that $\sum_i y_{i2}x_{i2} = 0.4$.
	Agent 1 now benefits from redistributing her sales to Asset 1. Specifically, agent $1$ maximizes her utility by selling according to $y_{11}=0.829974$ and $y_{12}=0.929974>0.925997$, yielding the utility $u_1(\hat{y})=118.541$.
\end{proof}

Recall \Cref{thm:pos1}, which stated that for even sales with $\alpha=1$, the social-welfare optimum is also a strong equilibrium. 
The following proposition shows that this is no longer true for non-even sales even if price impact is also linear.

\begin{prop}
\label{prop:noneven-no-sw-opt}
There is a fire sale game with non-even sales, $\alpha=1$, and linear
price impact such that there exists a social-welfare optimum $\yopt$ that is not
an equilibrium.
\end{prop}
\begin{proof}
	\begin{figure}[h]
		\centering
		\includegraphics{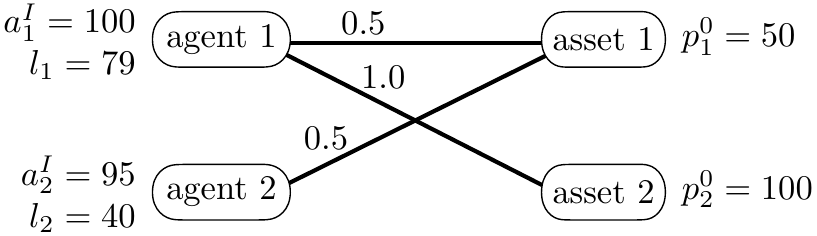}
		\caption{Fire sale game where the social-welfare optimum is not an equilibrium (\Cref{prop:noneven-no-sw-opt}). Let $\alpha=1$ and consider linear price impact.}
		\label{fig:noneven-no-sw-opt}
	\end{figure}
	Consider the game illustrated in \Cref{fig:noneven-no-sw-opt}. There are two agents and two assets and the game is defined by $a_1^I = 100$, $l_1= 79$, $x_{11}=0.5$, $x_{12}=1$ and $a_2^I=95$, $l_2 = 40$, $x_{21}=0.5$ with asset prices $p_1^0=50$, $p_2^0=100$ and $\lambda = 1.5$. Initially, agent $1$ violates the leverage constraint, since $\lev_1(y^1) = 1.5411$. The optimum is defined by
	\begin{align*}
		\max\quad        &  \sw(y) \\
		\text{s.t.\quad} & \lev_i(y)\leq \lambda & \forall i \in \{1,2\}\phantom{.} \\
		& y_{ij} \in [0,1] & \forall i,j \in \{1,2\}.
	\end{align*}
	It is attained at $\yopt_{11}\approx 0.8556$, $\yopt_{12}\approx 0.9$, and $\yopt_{21}=0.9444$ yielding the social welfare $\sw(\yopt)= u_1(\yopt)+u_2(\yopt)=133.5 + 77.5 = 211$. Agent $1$ can increase her utility to $u_1(y) \approx 138$ by deviating to $\hat y_1 := (\hat y_{11}\approx 0.5167,\, \hat y_{12}\approx 0.9889)$.
	Due to the decrease in price, Agent 2 violates the leverage constraint and, therefore, receives a utility of $u_2(y_2,\hat{y}_1)=-\infty$. As Agent 2 deviates to her best response $\hat{y}_{21}\approx 0.6212$ to $\hat{y}_1$, the social welfare is reduced to $\sw(\hat{y})\approx 138+69 = 207$. 
	Hence, the strategy profile maximizing the social welfare $\yopt$ forms no equilibrium.
\end{proof}

With regard to the stability of equilibrium dynamics, we show that improving-response cycles
can exist even for $\alpha=1$ and linear price impact, as the following example shows.
Recall that \Cref{ex:best-response-cycle-even}, which illustrates a best-response cycle for the case of even sales, assumed $\alpha < 1$ and strictly convex price impact.

\begin{example}
\label{ex:improving-response-cycle-noneven}
Consider the fire sale game in the proof of \Cref{prop:noneven-no-sw-opt}. Then the following is a cycle of improving responses:

\begin{center}
    \newcommand\centerinfty{\multicolumn{1}{c}{$-\infty$}}
\begin{tabular}{lllll}
     $y_{11}$ & $y_{12}$ & $y_{21}$ & $u_1$ & $u_2$
     \\\hline
    0.73541  & 0.749823 & 0.579619 & 112.42  & 71.4379\\
    0.73541  & 0.749823 & 0.802502 & 115.206 & 74.2239\\
    0.517823 & 0.933928 & 0.802502 & 130.897 & \centerinfty\\
    0.517823 & 0.933928 & 0.579619 & \centerinfty & 68.718
\end{tabular}%
\end{center}\hfill $\lhd$
\end{example}

\section{Discussion and Conclusions \label{sec:Discussion}}
We have studied price-mediated contagion from the perspective
of algorithmic game theory. The existence and the shape
of equilibria are heavily dependent on the assumptions regarding price
impact. For $\alpha=1$ or convex price impact, the set of equilibria forms
a lattice and the Pareto optima form strong equilibria.
However, agents face a twofold equilibrium
coordination problem: (1) while the maximal equilibrium is the social optimum, the minimal
equilibrium can be arbitrarily poor; (2) simplified best-response dynamics can take a long time to converge. 
This may lead to a delay in resolution and exacerbates a financial crisis. To help
alleviate this problem, a regulator may estimate the maximal equilibrium and help guide agents towards it.

For $\alpha \in (0,1)$ and concave price impact, an equilibrium need not exist. It is an interesting open problem to study the computational
complexity of deciding existence of equilibria in a given fire sale game. Another open problem is characterizing existence and computation of equilibria for non-even sales.

Several interesting directions arise for future work.
A prominent research direction are the economic implications of bailouts. In fire sale games, it can be beneficial to an agent
to transfer assets to another one without any compensation.

\begin{example}
Consider a fire sale game with two agents, one asset, and linear
price defined by $\alpha=1$, $\lambda=3.84615$, $p_{1}^{0}=1$,
and the following matrix of asset holdings:

\begin{center}
\begin{tabular}{lllll}
$x_{11}=0.8$  & $a_{1}^{I}=1000$  & $l_{1}=0$ &  & \tabularnewline
$x_{21}=0.2$  & $a_{2}^{I}=1$  & $l_{2}=0.9$ &  & \tabularnewline
\end{tabular}
\end{center}

Consider the Nash equilibrium $(1,0)$. The equity of Agent 1 is $1000.64$.
If now Agent 1 transfers a share of $0.1$ of Asset 1 to Agent
2, this leads to the new Nash equilibrium $(1,1)$, i.e., no agent
sells anything. The equity of Agent 1 in this new game is $1000.7$.
Agent 1 has thus improved her equity by giving away assets for free.
 \hfill $\lhd$
\end{example}

The bailout in the example benefits both agents and seems to be desirable. It appears prudent
to enable such bailouts, e.g., via rescue funds or mutual insurances put into place during normal
times to apply in times of crises. Our results demonstrate
that price impact can make bailouts economically efficient and incentive-aligned. 
There are many important further questions surrounding the implementation of bailouts.

\bibliographystyle{plainnat}
\bibliography{references}

\end{document}